\documentclass[journal,twocolumn,twoside]{IEEEtran}

\usepackage[dvipsnames]{xcolor}
\usepackage[cmex10]{amsmath}
\usepackage{graphicx,epic,eepic,epsfig,latexsym,amssymb,verbatim,color,revsymb}

\usepackage[english]{babel}
\usepackage[T1]{fontenc}
\usepackage{mathrsfs}
\usepackage{enumerate}
\usepackage{graphicx}
\usepackage{mathtools}
\usepackage[export]{adjustbox}
\usepackage{microtype}
\usepackage{MnSymbol}
\usepackage{lipsum}

\usepackage[colorlinks = true,
            linkcolor = red,
            urlcolor  = blue,
            citecolor = blue,
            anchorcolor = red]{hyperref}

\usepackage{orcidlink}

\newtheorem{theorem}{Theorem}
\newtheorem{corollary}[theorem]{Corollary}
\newtheorem{proposition}[theorem]{Proposition}
\newtheorem{lemma}[theorem]{Lemma}
 
\newtheorem{definition}[theorem]{Definition}  
\newtheorem{remark}[theorem]{Remark}  

\newcommand\qedsymbol{$\blacksquare$}

\newlength{\blank}
\newenvironment{proof-of}[1][{\hspace{-\blank}}]{{\medskip\noindent\textit{Proof~{#1}.\ }}}{\hfill\qedsymbol}
\newenvironment{proof}{{\medskip\noindent\textit{Proof.\ }}}{\hfill\qedsymbol}

\newcommand{\Tr}{{\operatorname{Tr}\,}}

\newcommand{\id}{{\operatorname{id}}}

\newcommand{\1}{\openone}
\newcommand{\ket}[1]{|#1\rangle}
\newcommand{\bra}[1]{\langle #1|}

\newcommand{\ketbra}[2]{|#1\rangle\!\langle #2|}
\newcommand{\proj}[1]{|#1\rangle\!\langle #1|}

\newcommand{\cD}{{\mathcal{D}}}
\newcommand{\cE}{{\mathcal{E}}}

\newcommand{\cS}{{\mathcal{S}}}

\newcommand{\nc}{\newcommand}
\nc{\rnc}{\renewcommand}
\nc{\beq}{\begin{equation}}
\nc{\eeq}{{\end{equation}}}
\nc{\beqa}{\begin{eqnarray}}
\nc{\eeqa}{\end{eqnarray}}
\nc{\lbar}[1]{\overline{#1}}
\nc{\avg}[1]{\langle#1\rangle}
\nc{\Rank}{\operatorname{Rank}}
\nc{\smfrac}[2]{\mbox{$\frac{#1}{#2}$}}
\nc{\tr}{\operatorname{Tr}}
\nc{\ox}{\otimes}

\nc{\RR}{{{\mathbb R}}}
\nc{\CC}{{{\mathbb C}}}
\nc{\FF}{{{\mathbb F}}}
\nc{\NN}{{{\mathbb N}}}
\nc{\ZZ}{{{\mathbb Z}}}
\nc{\PP}{{{\mathbb P}}}
\nc{\QQ}{{{\mathbb Q}}}
\nc{\UU}{{{\mathbb U}}}
\nc{\EE}{{{\mathbb E}}}

\newcommand{\wet}[1]{\widetilde{#1}}

\begin{document}

\title{A Rate-Distortion Perspective on \protect\\ Quantum State Redistribution} 

\author{Zahra~Baghali~Khanian\,\orcidlink{0000-0002-0892-7519} and Andreas~Winter\,\orcidlink{0000-0001-6344-4870}
\thanks{Date: 6 December 2024. 
Part of this work, focusing on the source coding aspects, 
has been presented at ISIT 2020 \cite{ISIT-paper} and has 
featured partially in the first author's PhD thesis \cite{ZBK_PhD_thesis}.
ZBK was supported by the DFG Cluster of Excellence 2111 (Munich Center for Quantum
Science and Technology, MCQST), and a Marie Sk{\l}odowska Curie Postdoctoral Fellowship 
of the European Commission. 
AW was supported by the Spanish MINECO and MICIN (projects PID2019-107609GB-I00
and PID2022-141283NB-I00) with the support of FEDER funds, 
by the Generalitat de Catalunya (project 2017-SGR-1127),
by the European Commission QuantERA grant ExTRaQT 
(Spanish MICIN project PCI2022-132965), 
by the Spanish MICIN with funding from European Union NextGenerationEU 
(PRTR-C17.I1) and the Generalitat de Catalunya, by the Spanish MTDFP 
through the QUANTUM ENIA project: Quantum Spain, funded by the European 
Union NextGenerationEU within the framework of the ``Digital Spain 
2026 Agenda'', by the Alexander von Humboldt Foundation, and by the 
Institute for Advanced Study of the Technische Universit\"at M\"unchen.
The authors thank M. Jupien for helpful demonstrations regarding
the advanced use of {\LaTeX}.}
\thanks{Z. B. Khanian is with the Munich Center for Quantum Science and Technology and Zentrum Mathematik, Technische Universit\"{a}t M\"{u}nchen, 85748 Garching, Germany. Email: zbkhanian@gmail.com.}
\thanks{A. Winter is with ICREA---Instituci\'o Catalana de Recerca i Estudis Avan\c{c}ats, Pg.~Lluis Companys, 23, 08010 Barcelona, Spain, and F\'{\i}sica Te\`{o}rica: Grup d'Informaci\'{o} Qu\`{a}ntica (GIQ), 
Departament de F\'{\i}sica, Universitat Aut\`{o}noma de Barcelona, 08193 Bellaterra (Barcelona), Spain. 
He is furthermore a Hans Fischer Senior Fellow with the Institute for Advanced Study, Technische Universit\"at M\"unchen,  Lichtenbergstra{\ss}e 2a, D-85748 Garching, Germany. Email: andreas.winter@uab.cat.}
}

\maketitle

\begin{abstract}
We consider a rate-distortion version of the quantum state redistribution task, 
where the error of the decoded state is judged via an additive distortion measure; 
it thus constitutes a quantum generalisation of the classical Wyner-Ziv problem.
The quantum source is described by a tripartite pure state shared between 
Alice ($A$, encoder), Bob ($B$, decoder) and a reference ($R$). 
Both Alice and Bob are required to output a system ($\widetilde{A}$ and $\widetilde{B}$, 
respectively), and the distortion measure is encoded in an observable on 
$\widetilde{A}\widetilde{B}R$. 
%
It includes as special cases most quantum rate-distortion problems considered in 
the past, and in particular quantum data compression with the fidelity 
measured per copy; furthermore, it generalises the well-known state merging and 
quantum state redistribution tasks for a pure state source, with per-copy fidelity, 
and a variant recently considered by us, where the source is an ensemble of pure 
states 
[ZBK {\&} AW, Proc. ISIT 2020, pp.~1858-1863 
and ZBK, PhD thesis, UAB 2020, arXiv:2012.14143].
We derive a single-letter formula for the rate-distortion function of 
compression schemes assisted by free entanglement. A peculiarity of the 
formula is that in general it requires optimisation over an unbounded 
auxiliary register, so the rate-distortion function is not readily 
computable from our result, and there is a continuity issue at zero distortion. 
However, we show how to overcome these difficulties in certain situations. 
%
\end{abstract}


\section{Introduction and setting}
\label{sec:introduction}
\IEEEPARstart{S}{ource} coding is for information theory as much a practical matter, as it is 
a fundamental paradigm to establish the amount of information in given data. 
Shannon's original model of block coding \cite{Shannon:IT}, giving operational meaning to 
the entropy, was subsequently generalized to situations with side information 
at the decoder \cite{SlepianWolf}, which gives an operational interpretation for 
the conditional entropy. In another direction, by considering more flexible 
error criteria, ``distortions'', instead of the rigid block error 
probability \cite{Shannon:RD,Berger:RD-book}, leads to a rate-distortion tradeoff 
characterized by the mutual information. Many other variations of source compression 
have been conceived, but to conclude our rapid review of classical source coding, 
we highlight only one more, the Wyner-Ziv problem of rate-distortion of a source
with correlated side information at the decoder \cite{WynerZiv}. 

Quantum Shannon theory has sought to emulate this approach by ``quantizing''
the preceding source coding problems, with the aim of gaining both a fundamental 
and operationally grounded understanding of quantum information. 
The first and most important among these is Schumacher's quantum source 
model and compression problem, whose optimal rate is given by the von Neumann 
entropy \cite{Schumacher1995,SchumacherJozsa1994,Barnum1996,Horodecki:source,Winter:PhD}. 
Compared to the classical case, quantum compression with side information turned out 
to have a surprisingly rich structure, see
\cite{Winter:PhD,Devetak2003,Ahn2006,negative-info,family,Abeyesinghe2009,Devetak2008_2,Yard2009,Oppenheim2008, ZK_cqSW_2018,ZK_cqSW_ISIT_2019}. 
%
On the other hand, rate-distortion theory has received scarce quantum attention 
over the years, and the results of \cite{Barnum:QRD,Winter:QRD,EA-RD-1,EA-RD-2} 
are not as complete as the classical theory.

Here we present and solve a quantum version of the Wyner-Ziv problem, with unlimited 
entanglement, for a distributed source and relative to a convex, additive distortion measure.
Concretely, we consider a pure state source $\ket{\psi}^{ABR}$, with $A$ Alice's 
register, $B$ Bob's and $R$ is a passive reference. 
Furthermore, let $\Delta:\cS(\wet{A}\wet{B}R) \rightarrow \RR$ be a
convex continuous real function on the set $\cS$ of the quantum states of the tripartite 
system $\wet{A}\wet{B}R$. For block length $n$, the source is the i.i.d. extension 
$\psi^{A^nB^nR^n} = \left(\psi^{ABR}\right)^{\ox n}$, and the distortion measure is 
extended to $n$ systems as 
\begin{equation}
  \Delta^{(n)}\left(\rho^{\wet{A}^n\wet{B}^nR^n}\right) 
              = \frac{1}{n} \sum_{i=1}^n \Delta\left(\rho^{\wet{A}_i\wet{B}_iR_i}\right), 
\end{equation}
where on the right hand side 
$\rho^{\wet{A}_i\wet{B}_iR_i} = \Tr_{[n]\setminus i} \rho^{\wet{A}^n\wet{B}^nR^n}$
is the reduced state on the $i$-th systems $\wet{A}_i\wet{B}_iR_i$ 
(partial trace over all other systems). 

An important special case, considered in earlier approaches to quantum rate-distortion
\cite{Winter:QRD,EA-RD-1,EA-RD-2}, is that $\Delta(\rho) = \tr\rho\Delta$ for a
selfadjoint distortion observable $\Delta$ and $\Delta^{(n)}(\rho) = \tr\rho\Delta^{(n)}$, 
where
\[
  \Delta^{(n)} = \frac{1}{n} \sum_{i=1}^n \1^{\ox i-1} \ox \Delta \ox \1^{\ox n-i}.
\]
We note that $\Delta$ is used to denote an observable, whereas $\Delta(\cdot)$ denotes a scalar
function, which in the case of the distortion observable $\Delta$ is given by the trace function. 
This justifies in our view the notation by the same letter, while concretely 
there is little danger of confusion. 
Distortion operators are enough to describe classical rate-distortion 
functions \cite{Shannon:RD,Berger:RD-book}, 
but the theory goes through in the generality of the above convex additive functions. 

\begin{figure*}[ht]
  \centering
 \includegraphics[width=0.81\textwidth]{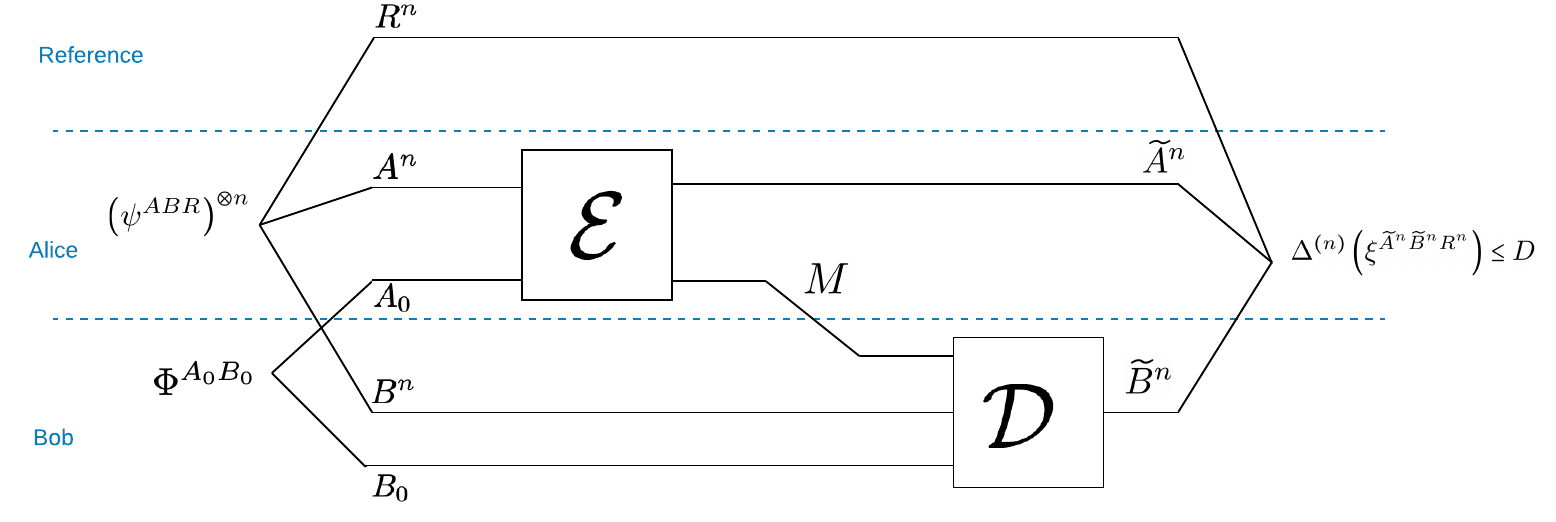}
  \caption{Communication diagramme of the entanglement-assisted rate-distortion state redistribution task: 
           $R$ is a passive reference, $A$ and $\wet{A}$ are Alice's input and output systems, 
           and $B$ and $\wet{B}$ are Bob's input and output systems, respectively.}
  \label{fig:QSI-RD}
\end{figure*}

With these data, an entanglement-assisted compression scheme of block length $n$ 
and bound on the distortion $D \in \RR$ consists of an entangled state
$\Phi^{A_0B_0}$, w.l.o.g. pure, and a pair of CPTP maps $\cE_n:A^nA_0 \rightarrow \wet{A}^nM$
and $\cD:MB_0B^n \rightarrow \wet{B}^n$ (see Fig.~\ref{fig:QSI-RD}), such that the output state 
\begin{equation}\begin{split}
  \label{eq: xi output state}
  &\xi^{\wet{A}^n\wet{B}^nR^n} \\
  &\phantom{==}
   = (\cD_n \ox \id_{\wet{A}^nR^n})\circ(\cE_n \ox \id_{B_0B^nR^n}) \left(\psi^{\ox n} \ox \Phi^{A_0B_0}\right)
\end{split}\end{equation}
satisfies the distortion constraint
\begin{equation}
  \Delta^{(n)}\left(\xi^{\wet{A}^n\wet{B}^nR^n}\right) 
     = \frac{1}{n} \sum_{i=1}^n \Delta\left(\xi^{\wet{A}_i\wet{B}_iR_i} \right) \leq D. 
\end{equation}
The rate of the code is simply $\frac{1}{n}\log|M|$, i.e. the number of qubits sent 
per source system. 

We say that a qubit rate $R_Q$ is \emph{asymptotically achievable with 
asymptotic distortion} $D \in \RR$, 
if there exists a sequence of codes $\{(\mathcal{E}_n,\mathcal{D}_n)\}_n$ such that 
\begin{align*}
    \Delta^{(n)}\left(\xi^{\wet{A}^n\wet{B}^nR^n}\right) \leq D + \delta_n \quad\text{and}\quad
    \frac{1}{n} \log |M| \leq R_Q + \eta_n, 
\end{align*}
for sequences $\delta_n \rightarrow 0$ and $\eta_n \rightarrow 0$ as 
$n \rightarrow \infty$.
The \emph{rate-distortion function} is defined as 
\begin{equation*}
    Q_{ea}(D) \coloneqq \inf \left\{R_Q : (R_Q,D) \text{ is achievable} \right\}, 
\end{equation*}
where the subscript `ea' reminds that the codes are
assisted by entanglement. 

We stress that $A$, $B$ and $R$ are arbitrary quantum systems here, and so 
are $\wet{A}$ and $\wet{B}$: the latter need not bear any relation to $A$ and $B$,
their names are chosen entirely as a reminder that `$A$'s belong to Alice
(compressor/sender) and `$B$'s belong to Bob (receiver/decoder). 

Note that for very small $D$, there may be no codes with distortion $D$,
and then $Q_{ea}(D)=+\infty$ by convention. Once codes exist, $Q_{ea}(D)$ is a non-negative 
real number, and for sufficiently large $D$, for example $D \geq \max_\rho \Delta(\rho)$,
$Q_{ea}(D)=0$ because every pair of maps is an eligible code. 

\medskip
\begin{remark}
In contrast to other previous work \cite{EA-RD-2}, which imposed a
distinction between data to be compressed and side information, we 
think that our present model is both simpler and more natural, by applying a 
global distortion measure jointly to Alice's and Bob's parts of the 
output, as well as to the reference. 
\end{remark} 

\medskip
In the rest of the paper, we present our main result in Section \ref{sec:main}, 
which is a single-letter characterization of the rate-distortion function, 
for general sources and arbitrary convex and continuous distortion measures: we first use quantum state redistribution (QSR) to build 
a protocol giving us an achievable rate, and then show that it is essentially 
optimal. Then, in Section \ref{sec:special}, we discuss a number 
of special cases of the considered scenario, showing how the rate-distortion 
setting generalizes all sorts of conventional quantum source coding problems, 
some of which have appeared in the previous literature. The original source coding 
problems are recovered, after a fashion, in the limit of zero (per-copy) error. 
We conclude with a discussion of the result and open problems in 
Section \ref{sec:discussion}.

\medskip
\textbf{Notation and basic facts.} 
Quantum systems are associated with (in this paper: finite dimensional) 
Hilbert spaces $A$, $R$, \ldots, whose dimensions are denoted by $|A|$, $|R|$, \ldots, respectively. 
We identify states on a system $A$ with their density operators, $\cS(A)$, 
which is the set of all positive semidefinite matrices with unit trace. 
We use the notation $\phi= \ketbra{\phi}{\phi}$ as the density operator of
the pure state vector $\ket{\phi} \in A$. 
 
The von Neumann entropy is $S(\rho) = - \Tr\rho\log\rho$, 
$\log$ by default being the binary logarithm.
The conditional entropy and the conditional mutual information, $S(A|B)_{\rho}$ and $I(A:B|C)_{\rho}$,
respectively, are defined in the same way as their classical counterparts: 
\begin{align*}
  S(A|B)_{\rho}   &= S(AB)_\rho-S(B)_{\rho}, \text{ and} \\ 
  I(A:B|C)_{\rho} &= S(A|C)_\rho-S(A|BC)_{\rho} \\
                  &= S(AC)_\rho+S(BC)_\rho-S(ABC)_\rho-S(C)_\rho.
\end{align*}

The fidelity between two states $\rho$ and $\xi$ is defined as 
\(
 F(\zeta,\xi) = \left\|\sqrt{\zeta}\sqrt{\xi}\right\|_1 
              = \Tr \sqrt{\zeta^{\frac{1}{2}} \xi \zeta^{\frac{1}{2}}},
\) 
with the trace norm $\|X\|_1 = \Tr|X| = \Tr\sqrt{X^\dagger X}$.
If one of the two states is pure, $F(\zeta,\xi)^2=\Tr \zeta\xi$.
In general, the fidelity relates to the trace distance in the following well-known 
way \cite{Fuchs1999}:
\begin{equation}
  \label{eq:FvdG}
  1-F(\zeta,\xi) \leq \frac12\|\zeta-\xi\|_1 \leq \sqrt{1-F(\zeta,\xi)^2}.
\end{equation}

As we consider information theoretic limits, we have occasion to 
refer to many isomorphic copies of a single system, say $A$, which
are always referred to by the same capital letter with a running 
index, i.e. $A_1$, $A_2$, \ldots, $A_n$; a block (tensor 
product) of the first $n$ of these systems is 
written $A^n = A_1 A_2\cdots A_n = A_1 \ox \cdots \ox A_n$.
More generally for a set $I\subset\mathbb{N}$ of indices, $A_I = \bigotimes_{i\in I} A_i$.
We use the combinatorial shorthand $[n] = \{1,2,\ldots,n\}$, so that $A^n = A_{[n]}$.

\section{Single-letter characterization\protect\\ of the rate-distortion function}
\label{sec:main} 
In this section, we solve the quantum rate-distortion problem introduced above 
(depicted in Fig.~\ref{fig:QSI-RD}).
First, we construct a protocol for a certain achievable rate, coming directly 
from quantum state redistribution (QSR); after that, we show the converse. 
%
%
QSR is a quantum compression protocol where both encoder and decoder have access to side information. We introduce this protocol more in subsection~\ref{subsection:QSR}.

\subsection{An achievable rate from QSR}
Assume that we have two 
CPTP maps $\cE_0:A \rightarrow \wet{A}Z$ and $\cD_0:BZ \rightarrow \wet{B}$
such that for 
\begin{align}
  \label{eq:xi_final_state}
  \xi^{\wet{A}\wet{B}R} = (\cD_0 \ox \id_{\wet{A}R})\circ(\cE_0 \ox \id_{BR})\psi^{ABR},
\end{align}
it holds $\Delta(\xi) \leq D$. Then, an achievable
asymptotic rate for distortion $D$ is given by $R = \frac12 I(Z:R|B)_\varphi$, 
with the state after the action of $\cE_0$, $\varphi^{\wet{A}ZBR} = (\cE_0 \ox \id_{BR})\psi^{ABR}$. 

\begin{proof}
We prove the achievability of the above rate as follows.
Purify $\cE_0$ to a Stinespring isometry $U:A\hookrightarrow \wet{A}ZW$ \cite{Stinespring1955},
so after applying it to the source we have the pure state
$\ket{\varphi}^{\wet{A}WZBR} = (U\ox\1_{BR})\ket{\psi}^{ABR}$. On block length $n$, 
use QSR, assisted by suitable entanglement, as a subroutine, to send $Z^n$ 
from Alice to Bob, with $\wet{A}^nW^n$ as Alice's side information and $B^n$ as Bob's. 
The block trace distance error of the QSR protocol goes to $0$ as $n\rightarrow \infty$,
so we get distortion $\leq D+o(1)$, using the continuity of $\Delta$. 
The rate, which is due to QSR, is $\frac12 I(Z:R|B)_\varphi$ \cite{Devetak2008_2,Yard2009,Oppenheim2008}.
\end{proof}

\medskip
This coding theorem motivates the introduction of the following single-letter function,
\begin{equation}\begin{split}
  \label{eq:rd-function-almost}
  Q'(D) &:= \inf_{\cE_0, \cD_0} \frac12 I(Z:R|B)_\varphi \text{ s.t. } \cE_0:A \rightarrow \wet{A}Z \text{ and } \\
        &\phantom{=======}
         \cD_0:BZ \rightarrow \wet{B}  \text{ CPTP, and } \Delta(\xi) \leq D,
\end{split}\end{equation}
where $\xi$ is defined in Eq.~(\ref{eq:xi_final_state}), and the conditional mutual information is with respect to the state $\varphi^{\wet{A}ZBR}= (\cE_0 \ox \id_{BR})\psi^{ABR}$. 
With this notation, what we have just proved amounts to
\begin{equation}
  Q_{ea}(D) \leq \lim_{D'\rightarrow D+} Q'(D'). 
\end{equation}
Since $Q'$ is monotonically non-increasing with $D$, the latter limit from the right 
is also a supremum, equal to $\displaystyle{\sup_{D'>D} Q'(D')}$. 

Before we go on, we analyze first some mathematical properties of the new
function. Note that a major difficulty, both practically and for the theoretical 
development, is the unbounded nature of the auxiliary system $Z$. 
Define
\begin{equation}\begin{split}
  \label{eq: D_0}
  D_0 &:= \inf D \text{ s.t. } Q'(D) < +\infty \\
      & = \inf D \text{ s.t. } \exists\, \cE_0,\cD_0\ \Delta(\xi) \leq D. 
\end{split}\end{equation}
By definition, $Q'(D) = +\infty$ for all $D<D_0$ and $Q'(D)$ is finite for 
all $D>D_0$. Because of the dimensionality issue, $Q'(D_0)$ may or may not be 
finite. 

\medskip
\begin{lemma}
  \label{lemma:Q-convex}
  On $[D_0,\infty)$, $Q'$ is a monotonically non-increasing, convex function of $D$.
  Consequently, on the open interval $(D_0,\infty)$ it is also continuous. 
\end{lemma}

\begin{proof}
The monotonicity was already remarked to follow from the definition. 
For the convexity, we verify Jensen's inequality, that is we start with maps 
$\cE_1,\cD_1$ eligible for distortion $D_1$ [as defined in Eq.~(\ref{eq:rd-function-almost})], 
and $\cE_2,\cD_2$ eligible for distortion $D_2$,
and $0\leq p \leq 1$. By embedding into larger Hilbert spaces if necessary, we
can w.l.o.g. assume that the maps act on the same systems for $i=1,2$.
We define the following two maps:
\begin{align*}
  \cE(\rho^A) &:= p \cE_1(\rho) \ox \proj{1}^{Z'} + (1-p) \cE_2(\rho) \ox \proj{2}^{Z'}, \\
  \cD(\sigma^{BZ}) &:= \cD_1(\bra{1}^{Z'} \sigma \ket{1}^{Z'}) + \cD_2(\bra{2}^{Z'} \sigma \ket{2}^{Z'}).
\end{align*}
They evidently realise the output state $\xi = p\xi_1 + (1-p)\xi_2$, where $\xi_i = (\cD_i \ox \id_{\wet{A}R})\circ(\cE_i \ox \id_{BR})\psi^{ABR}$ for $i=1,2$. Hence by convexity the distortion is bounded as
$\Delta(\xi) \leq p \Delta(\xi_1) + (1-p) \Delta(\xi_2) \leq pD_1+(1-p)D_2=D$. 
Thus, 
\[\begin{split}
  Q'(D) &\leq\frac{1}{2} I(ZZ':R|B)_\xi \\
        &=    \frac{p}{2} I(Z:R|B)_{\xi_1} + \frac{(1-p)}{2} I(Z:R|B)_{\xi_2},
\end{split}\]
and taking the infimum over maps $\cE_i,\cD_i$ shows convexity. 

The continuity statement follows from a mathematical folklore fact, stating that 
any real-valued function that is convex on an interval, is continuous on the 
interior of the interval, cf.~\cite[Prop.~2.17]{Barbu-book}.
\end{proof}

\medskip
This lemma shows that the only possible discontinuity of $Q'$ is at $D_0$, and so we 
are motivated to define its right-continuous extension, which differs from $Q'$ only 
possibly at $D_0$:
\begin{equation}\label{overline Q(D)}
  \overline{Q}(D) := \sup_{D'>D} Q'(D')
                   = \begin{cases}
                       +\infty              & \text{ if } D < D_0, \\
                       \sup_{D'>D_0} Q'(D') & \text{ if } D = D_0, \\
                       Q'(D)                & \text{ if } D > D_0.
                     \end{cases}
\end{equation}
Our achievability result from the beginning of the present section can now be 
expressed more concisely as follows. 

\medskip
\begin{proposition}
\label{prop:achieve}
For any source $\psi^{ABR}$ and any convex distortion measure $\Delta(\cdot)$,
it holds for all distortion values $D$ that 
\[
  Q_{ea}(D) \leq \overline{Q}(D).
\]
(Note that this is trivially true for $D < D_0$, as then the right hand side,
and as we shall see also the left hand side, is $+\infty$.)
\end{proposition}

\subsection{Main result}
%
\begin{theorem}
\label{thm:main}
For any source state $\psi^{ABR}$ and any convex distortion measure $\Delta$,
it holds for all distortion values $D$ that 
\[
  Q_{ea}(D) = \overline{Q}(D).
\]
\end{theorem}
\begin{proof}
In light of Proposition \ref{prop:achieve}, stating that
$Q_{ea}(D) \leq \overline{Q}(D)$, we only have to prove the opposite 
inequality, i.e.~$Q_{ea}(D) \geq \overline{Q}(D)$, in other words the converse. 

Towards this end, consider a block length $n$ code of distortion
$\Delta^{(n)}\left(\xi^{\wet{A}^n\wet{B}^nR^n}\right) \leq D+\delta$ for the output state defined in Eq.~(\ref{eq: xi output state}). The number 
of qubits, $\log|M|$, can be lower bounded as follows, with respect to 
the encoded state $\sigma^{MB_0\wet{A}^nB^nR^n} = (\cE\ox\id_{B_0B^nR^n})(\psi^{A^nB^nR^n}\ox\Phi^{A_0B_0})$:
\begin{align}
  2\log|M| &\geq 2 S(M) \nonumber\\
           &\geq I(M:R^n|B^nB_0) \nonumber\\
           &=    I(MB_0:R^n|B^n) - \underbrace{I(B_0:R^n|B^n)}_{=0} \nonumber\\
           &=    I(Z:R^n|B^n)                           \qquad [\text{with } Z \equiv MB_0] \nonumber\\
           &=    \sum_{i=1}^n I(Z:R_i|B^nR_{<i}) 
                       + \sum_{i=1}^n \underbrace{I(R_{<i}B_{[n]\setminus i}:R_i|B_i)}_{=0} \nonumber\\
           &=    \sum_{i=1}^n I(ZR_{<i}B_{[n]\setminus i}:R_i|B_i)                          \nonumber\\
           &\geq \sum_{i=1}^n I(Z_i:R_i|B_i),           \qquad [\text{with } Z_i \equiv Z B_{[n]\setminus i}]
                  \label{eq:lower-bound-1}
\end{align}
where in the first two inequalities we use standard entropy inequalities;
the equation in the third line is due to the chain rule, and the second 
conditional information is $0$ because $B_0$ is independent of $B^nR^n$;
the fourth line introduces a new register $Z$, noting that 
the encoding together with the entangled state defines a CPTP map 
$\cE_0:A^n \rightarrow \wet{A}^nZ$, via $\cE_0(\rho) = (\cE\ox\id_{B_0})(\rho\ox\Phi^{A_0B_0})$;
in the fifth we use the chain rule iteratively, and in the second term we introduce 
each summand is $0$ because for all $i$, $R_{<i}B_{[n]\setminus i}$ is independent of $R_iB_i$;
in the sixth line we use again the chain rule for all $i$, and in the
last line strong subadditivity (data processing).

\medskip\noindent
For the $i$-th copy $\psi^{A_iB_iR_i}$, now define maps 
$\cE_i:A_i \rightarrow \wet{A}_iZ_i$ and $\cD_i:B_iZ_i \rightarrow \wet{B}_i$,
as follows:

\begin{itemize}
\setlength{\itemindent}{3mm}
\item[$\cE_i$:] Alice tensors her system $A_i$ with a dummy state $\psi^{\ox[n]\setminus i}$ and
      with $\Phi^{A_0B_0}$ (note that all systems are in her possession). 
      Then she applies $\cE:A^nA_0 \rightarrow \wet{A}^nM$, and sends 
      $Z_i := M B_0 B_{[n]\setminus i}$ to Bob, while keeping $\wet{A}_i$.
      Everything else, i.e. $R_{[n]\setminus i} \wet{A}_{[n]\setminus i}$, is trashed.
\item[$\cD_i$:] Bob applies $\cD$ to $Z_iB_i = M B_0 B^n $ and keeps 
      $\wet{B}_i$, trashing the rest $\wet{B}_{[n]\setminus i}$.
\end{itemize}
By definition, the output state 
\[
  \zeta^{\wet{A}_i\wet{B}_iR_i} = (\cD_i \ox \id_{\wet{A}_iR_i})\circ(\cE_i \ox \id_{B_iR_i})\psi^{A_iB_iR_i}
\]
equals $\xi^{\wet{A}_i\wet{B}_iR_i}=\Tr_{[n]\setminus i} \xi^{\wet{A}^n\wet{B}^nR^n}$, 
and with the $i$-th letter distortion $D_i := \Delta\left(\zeta^{\wet{A}_i\wet{B}_iR_i}\right)$ we have
\[\begin{split}
  D+\delta &\geq \Delta^{(n)}\left(\xi^{\wet{A}^n\wet{B}^nR^n}\right)     \\
           &=    \frac{1}{n} \sum_{i=1}^n \Delta\left(\xi^{\wet{A}_i\wet{B}_iR_i}\right) \\
           &=    \frac{1}{n} \sum_{i=1}^n \Delta\left(\zeta^{\wet{A}_i\wet{B}_iR_i}\right) 
            =    \frac{1}{n} \sum_{i=1}^n D_i.
\end{split}\]
Thus, we obtain, with respect to the states $(\cE_i \ox \id_{B_iR_i})\psi^{A_iB_iR_i}$
for $i=1\ldots,n$, 
\begin{align}\label{eq: converse Q'}
  \frac{1}{n}\log|M| &\geq \frac{1}{n} \sum_{i=1}^n \frac12 I(Z_i:R_i|B_i) \nonumber\\
                     &\geq \frac{1}{n} \sum_{i=1}^n Q'(D_i)               \nonumber\\
                     &\geq Q'\left( \frac{1}{n} \sum_{i=1}^n D_i\right)  \nonumber\\
                     &\geq Q'(D+\delta), 
\end{align}
continuing from Eq.~(\ref{eq:lower-bound-1}), then by definition of $Q'(D_i)$ since 
the pair $(\cE_i,\cD_i)$ results in distortion $D_i$, in the next line 
by convexity and finally by monotonicity of $Q'$ (Lemma \ref{lemma:Q-convex}).

Since this has to hold for all $\delta > 0$ and in the limit $n\rightarrow\infty$, 
the claim follows.
\end{proof}

\medskip
\begin{remark}
The real problem with Theorem \ref{thm:main}, and the formula (\ref{eq:rd-function-almost}), 
is that while the rate-distortion function on the face of it is single-letter, it is
still not necessarily computable, because of the infimum over CPTP maps $\cE_0:A \rightarrow \wet{A}Z$ 
and $\cD_0:BZ \rightarrow \wet{B}$, with -- crucially -- unbounded quantum register $Z$. 

With a bounded $|Z|$, the domain of optimization would become compact, and 
this would not only make $Q'(D)$ computable (in the sense that it can be approximated 
to arbitrary degree), and in fact a minimum, hence itself a continuous function, 
but also $D_0$ would be computable, and we would get $\overline{Q} \equiv Q'$. 

Without this information, and we have no evidence of finiteness or required
infinity either way, in general, the rate-distortion function is only a 
formal expression, and shares the issue of computability or approximability 
with an astonishing number of other, similar capacity formulas in quantum 
Shannon theory: 
the entangling power of a bipartite unitary \cite{ent-power:U}, 
the symmetric side-channel assisted quantum capacity \cite{Q_ss} 
and the analogous private capacity \cite{P_ss}, 
the squashed entanglement \cite{squashed} 
the so-called conditional entanglement of mutual information (CEMI) \cite{CEMI}, 
and the quantum information bottleneck function \cite{QIBM}.
\end{remark}

\medskip
In the rest of the paper, we will show how this theorem permits a new 
view of various quantum source coding problems that have been considered 
in the literature previously. In all these cases, this rests on writing 
the pure state or the ensemble fidelity (per-copy) of a coding scheme as a 
distortion in the above sense.

\section{Entanglement-assisted source coding emerging in the limit of unit per-copy fidelity}
\label{sec:special}
In this section we are going to specialise the above general theory to 
the traditional setting of quantum source compression, where the distortion 
measure is the \emph{infidelity of decoding}, i.e. one minus the fidelity 
(squared) between the decoded state and an ideal state. This means that 
in all the distortion functions defined in the sequel, $D_0=0$. 
Note, however, that unlike the usual setting of Schumacher's data compression, 
we allow for potentially unlimited entanglement, which affects the rate in 
certain scenarios.

\subsection{Schumacher's quantum data compression with an entanglement fidelity criterion}
\label{subsec:Schumacher-pure}
In \cite{Schumacher1995,SchumacherJozsa1994}, quantum source coding is 
described with a pure state $\ket{\psi}^{AR}$ for the source, so that $B$ is 
trivial (one-dimensional) and so is $\wet{A}$, while $\wet{B} = \widehat{A} \simeq A$. 
The use of the (block) fidelity as success criterion of the code there, 
would correspond to the distortion measure $1-F(\xi,\psi)$, 
which would be eligible, being convex and continuous in the state. 
Here, we will however consider $\Delta(\xi) = 1-F(\xi,\psi)^2$, because it comes
from a distortion operator, $\Delta = \1 - \psi^{\widehat{A}R}$, which will suit us better 
in the later developments. Note that for regular source coding, this is not 
an important change, since there anyway the focus is on $F(\xi,\psi) \approx 1$; 
to be precise, for the $n$-fold i.i.d. repetition $\psi^{\ox n}$ and the 
$n$-system output state $\xi^{\widehat{A}^nR^n}$, one demands 
$F(\xi,\psi^{\ox n}) \approx 1$ in \cite{Schumacher1995,SchumacherJozsa1994}. 
Under the present rate-distortion perspective, however, we consider the 
weaker (implied) criterion $\Delta^{(n)}(\xi) = \tr\xi\Delta^{(n)} \approx 0$. 
Of course, rate-distortion theory makes good sense of all values of $D$, but 
we shall focus on the small ones to preserve the relation with source coding.
Schumacher's date compression implies that for all $D\geq 0$, 
$Q_{ea}(D) \leq Q_{ea}(0) \leq S(A)_\psi$. The latter bound is actually an 
equality, as it can be seen as follows (cf.~\cite{Horodecki:source}). Consider 
a $D>0$, and consider pairs of CPTP maps $\cE_0$ and $\cD_0$ eligible for 
$\overline{Q}(D)$, then Theorem~\ref{thm:main} implies the following  converse bound 
considering per-copy fidelity:
\begin{align}
  \overline{Q}(D) &=\inf \frac12 I(Z:R)_\varphi           \nonumber\\
        &\geq \inf \frac12 I(\widehat{A}:R)_\xi \nonumber\\
        &\geq \frac12 I(A:R)_\psi - 2\sqrt{D}\log|R| - g\left(\!\sqrt{D}\right),
         \label{eq:Schumi-lower}
\end{align}
where the first line is by definition, the second invoking data processing, 
and the last one by first observing that by Eq.~(\ref{eq:FvdG}), 
$\frac12 \|\xi-\psi\|_1 \leq \sqrt{D}$ and then using the Alicki-Fannes 
continuity bound for the conditional entropy \cite{Alicki2004} in the form 
given in \cite{Winter2016}: 
for two states with $\frac12 \left\|\rho^{UV}-\sigma^{UV} \right\|_1 \leq \delta$, 
\begin{equation}
  \label{thm:AFW}
  \left| S(U|V)_\rho - S(U|V)_\sigma \right| \leq 2\delta\log|U| + g(\delta),
\end{equation}
with $g(\delta) = (1+\delta)\log(1+\delta)-\delta\log\delta$. 

Thus, from the bound $ Q_{ea}(0) \leq S(A)_\psi$ and Eq.~(\ref{eq:Schumi-lower}), 
by letting $D\rightarrow 0$, we get
$Q_{ea}(0) = \overline{Q}(0) =\frac12 I(A:R)_\psi = S(A)_\psi$. This is the 
same rate as Schumacher's \cite{Schumacher1995,SchumacherJozsa1994}, but we stress 
that we get the optimality (lower bound) under the weaker assumption of the 
per-letter fidelity, rather than the block fidelity, being close to $1$.


\subsection{Schumacher's quantum data compression for an ensemble source}
\label{subsec:Schumacher}
Schumacher \cite{Schumacher1995} also introduced another model of the quantum 
source, as an ensemble $\{p(x),\ket{\psi_x}^A\}$, where $x$ ranges over a 
discrete set. One can of course describe this kind of source by a cq-state
$\omega^{AR} = \sum_x p(x) \proj{\psi_x}^A \ox \proj{x}^R$, but the rate-distortion 
setting allows to do it differently: not by changing the source state, which
remains the pure state $\ket{\psi}^{AR} = \sum_x \sqrt{p(x)} \ket{\psi_x}^A \ox \ket{x}^R$, 
but instead with a different distortion operator:
\begin{equation}\label{eq:D ensemble Schumacher}
  \Delta = \1 - \sum_x \proj{\psi_x}^{\widehat{A}} \ox \proj{x}^R.
\end{equation}
This does not change the task and the distortion measure at hand, but it 
allows us to use the framework introduced above. 
For this ensemble source, the output state of the composite system is 
\[\begin{split}
  \xi^{\widehat{A}^n  R^n} 
  &= (\cD_n \ox \id_{R^n})\circ(\cE_n \ox \id_{B_0 R^n})\left({(\psi^{AR}})^{\ox n} \ox \Phi^{A_0B_0}\right),
\end{split}\]
and the output state of the $i$-th system is 
\begin{align*}
    \xi^{\widehat{A}_i R_i} = \Tr_{[n]\setminus i} \xi^{\widehat{A}^n  R^n}=\sum_{x_i,x_i^{'}} \sqrt{p(x_i)p(x_i^{'})} \xi_{x_i ,x_i^{'}}^{\widehat{A}_i  R_i^{'}} \ox \ketbra{x_i}{{x'}_i}^{X_i},
\end{align*}
where 
$\xi_{x_i,x_i^{'}}^{\widehat{A}_i R_i^{'}}=\sum_{x_{[n] \setminus i}}  p(x_{[n] \setminus i}) \xi_{x^n, x'_i x_{[n] \setminus i}}^{\widehat{A}_i  R_i^{'}}$.
%
%
Measuring the distortion with the distortion operator of Eq.~(\ref{eq:D ensemble Schumacher}) is equivalent to measuring  per-copy fidelity for the output state
$\xi^{\widehat{A}_i R_i}$: 
\begin{align*}
    \tr \xi^{\widehat{A}_i  R_i} \Delta 
    &= 1 - \sum_{x_i} p(x_i) \tr \xi_{x_i ,x_i}^{\widehat{A}_i  {R'}_i}\psi_{x_i}^{\widehat{A}_i  {R'}_i} \\
    &= 1 - \sum_{x_i} p(x_i) F\left(\xi_{x_i ,x_i}^{\widehat{A}_i  {R'}_i},\psi_{x_i}^{\widehat{A}_i  {R'}_i}\right)^2.
\end{align*}
%


The optimal entanglement-assisted compression rate for this ensemble source is found in 
\cite{ZK_Eassisted_ISIT_2019} to be $\frac12 (S(A)_{\omega}+S(A|Y)_{\omega})$ where the decodability criterion is block fidelity. This rate is with respect to the following modified source defined as 
\begin{align}\label{modified Schumacher source}
  \omega^{AYR}:=\sum_{x} p(x)   \proj{\psi_x}^A  \otimes \proj{y(x)}^Y \otimes \proj{x}^R, 
\end{align}
where the register $Y$ stores the corresponding orthogonal subspaces for the signals $\{\psi_x^A \}$. For example, if signals $\psi_1^A$ and $\psi_2^A$ are orthogonal to signals $\psi_3^A$ and $\psi_4^A$, the variables $y(1)=y(2)$ and $y(3)=y(4)$ denote two underlying orthogonal subspaces. Since block fidelity implies per-copy
fidelity, the rate $\frac12 (S(A)_{\omega}+S(A|Y)_{\omega})$ is achievable with per-copy fidelity as well. Notice that 
\begin{align}
    \frac12 (S(A)_{\omega}+S(A|Y)_{\omega}) = \frac12 I(A:R)_{\psi'}, \nonumber
\end{align}
where the mutual information is with respect to the modified pure state 
$\ket{\psi'}^{AYR} := \sum_x \sqrt{p(x)} \ket{\psi_x}^A \ox \ket{y(x)}^A \ox \ket{x}^R$.

The converse bound for the above rate considering per-copy fidelity is contained in 
Corollary~\ref{cor:converse for modified Schumacher ensemble} below: 
\[
  Q_{ea}(0) \geq \overline{Q}(0) \geq \frac12 (S(A)_{\omega}+S(A|Y)_{\omega}),
\]
that is to say the same as that found in \cite{ZK_Eassisted_ISIT_2019}, but as 
before we stress that here it holds under the weaker per-copy fidelity. 
Note that the lower bound in Eq.~(\ref{eq:Schumi-lower}) is not valid here since in the last line we
use the fact that the decoded state on systems $\widehat{A} R$ is very close, in trace distance, to the pure source on systems $AR$. However, for the ensemble source, the decoded
state is close to the original ensemble and not the purification of the ensemble, therefore,
the lower bound on Eq.~(\ref{eq:Schumi-lower}) does not hold in this case.

\subsection{Quantum state redistribution}
\label{subsection:QSR}
To recover QSR itself, but with a per-letter fidelity criterion, 
we replace $A$ by the bipartite system $AC$, $\wet{A} = \widehat{C} \simeq C$
and $\wet{B} = \widehat{A}\widehat{B} \simeq AB$. The source is given by the pure state 
$\ket{\psi}^{ACBR}$, where $A$ and $C$ are initially with Alice and $B$ with Bob, 
and at the end $A$ changes hands from Alice to Bob, while $C$ and $B$ 
remain in place. 
The distortion operator is $\Delta = \1-\psi^{\widehat{A}\widehat{C}\widehat{B}R}$,
so that the distortion per letter is $\Delta(\xi) = \tr\xi\Delta = 1-\tr\xi\psi = 1-F(\xi,\psi)^2$. 

Note that for a single system the criterion is the familiar fidelity (up to the 
square, that some authors put and others not), but for block length $n$ the usual 
criterion considered for QSR \cite{Devetak2008_2,Yard2009,Oppenheim2008}
is not the per-copy but the block fidelity, which is a stronger requirement. 
Nevertheless, the well-known coding theorems for QSR \cite{Devetak2008_2,Yard2009,Oppenheim2008}
imply that $\frac12 I(A:R|B)$ is an achievable rate for any distortion $D \geq 0$, since block fidelity implies per-copy fidelity, 
hence $Q_{ea}(D) \leq \frac12 I(A:R|B) = \frac12 I(A:R|C)$ for all $D\geq 0$. 

%
On the other hand, for $D \geq 0$, Theorem~\ref{thm:main} implies the converse bound  $Q_{ea}(D) \geq \overline{Q}(D)$ considering per-copy fidelity. Namely, for $D\geq 0$ and  pairs of CPTP maps $\cE_0$ and $\cD_0$ eligible for 
$\overline{Q}(D)$,  we  obtain
\begin{align}
  \overline{Q}(D) &\geq \inf \frac12 I(Z:R|B)_\varphi           \nonumber\\
        &\geq \inf \frac12 I(\widehat{A}:R|B)_\xi \nonumber\\
        &\geq \frac12 I(A:R|B)_\psi - 2\sqrt{D}\log|R| - g\left(\!\sqrt{D}\right),
         \label{eq:QSR-lower}
\end{align}
where the first line is by definition, the second invoking data processing, 
and the last one by first observing that by Eq.~(\ref{eq:FvdG}), 
$\frac12 \left\|\xi^{\widehat{A} \widehat{C} \widehat{B} R}-\psi^{ACBR} \right\|_1 \leq \sqrt{D}$ 
and then using the Alicki-Fannes 
continuity bound for the conditional entropy \cite{Alicki2004} in the form of Eq.~(\ref{thm:AFW})
given in \cite{Winter2016}.
This lower bound together with the upper bound  discussed above imply that in the limit of $D\rightarrow 0$, 
$Q_{ea}(0) = \overline{Q}(0)$ converges to $\frac12 I(A:R|B)_\psi$. 
As in Subsection \ref{subsec:Schumacher-pure}, we stress that the optimality 
statement yields the same rate as \cite{Devetak2008_2,Yard2009,Oppenheim2008}, 
but under the weaker assumption of per-letter fidelity, rather than block 
fidelity being close to $1$. 

An important special case of QSR is state merging, which is recovered 
for trivial (one-dimensional) side-information system $C$, that is the source is given by the pure state 
$\ket{\psi}^{ABR}$, and $\wet{A} = 1$
and $\wet{B} = \widehat{A} \widehat{B}\simeq AB$ are respectively Alice and Bob's decoded systems. As discussed above, we can conclude that for per-copy fidelity (distortion operator $\Delta = \1-\psi^{\widehat{A}\widehat{B}R}$), 
the optimal rate is $\frac12 I(A:R)_{\psi}$. 


\subsection{Ensemble quantum state redistribution}
\label{sec: ensemble QSR}
Analogous to the discussion of Schumacher's quantum source coding 
(Subsections \ref{subsec:Schumacher-pure} and \ref{subsec:Schumacher}), 
if we have a source ensemble $\left\{p(x),\ket{\psi_x}^{ACBR'}\right\}$, we 
can represent this by the qqqqc-state $\omega^{ACBR'X}=\sum_x p(x) \proj{\psi_x}^{ACBR'} \ox \proj{x}^{X}$. 

However, we can also define the pure state source
$\ket{\psi}^{ACBR} := \sum_x \sqrt{p(x)} \ket{\psi_x}^{ACBR'}\ket{x}^X$ 
and a distortion operator such that measuring the distortion for the pure 
source $\ket{\psi}^{ACBR}$ is equivalent to 
measuring the ensemble infidelity for the source $\omega^{ACBR'X}$. 
%
As before, replace $A$ by the bipartite system $AC$, $\wet{A} = \widehat{C} \simeq C$,
$\wet{B} = \widehat{A}\widehat{B} \simeq AB$, and $R=R'X$.
Then, the output state of the composite system is 
\[\begin{split}
  &\xi^{\widehat{A}^n \widehat{C}^n \widehat{B}^n R^n} \\
  &\phantom{=}
   = (\cD_n \ox \id_{\widehat{C}^nR^n})\!\circ\!(\cE_n \ox \id_{B_0B^nR^n}) \left({(\psi^{ACBR}})^{\ox n} \ox \Phi^{A_0B_0}\right),
\end{split}\]
and the output state of the $i$-th system is 
\begin{align*}
  \xi^{\widehat{A}_i \widehat{C}_i \widehat{B}_i R_i} 
    &= \Tr_{[n]\setminus i} \xi^{\widehat{A}^n \widehat{C}^n \widehat{B}^n R^n} \\
    &=\sum_{x_i,x_i^{'}} \sqrt{p(x_i)p(x_i^{'})} \xi_{x_i ,x_i^{'}}^{\widehat{A}_i \widehat{C}_i \widehat{B}_i R_i^{'}} \ox \ketbra{x_i}{{x'}_i}^{X_i},
\end{align*}
where 
$\xi_{x_i,x_i^{'}}^{\widehat{A}_i \widehat{C}_i \widehat{B}_i R_i^{'}}=\sum_{x^n \setminus x_i}  p(x_{[n] \setminus i}) \xi_{x^n, x'_i x_{[n] \setminus i}}^{\widehat{A}_i \widehat{C}_i \widehat{B}_i R_i^{'}}$.

%
%
Define the distortion operator (we consider the same distortion operator for all copies 
of the source, that is why in the following definition, we drop the index $i$) 
\begin{align}\label{eq:Delta ensemble QSR}
    \Delta = \sum_x \left(\1-\psi_x^{\widehat{A}\widehat{C}\widehat{B}R'}\right) \ox \proj{x}^X,
\end{align}
so that the distortion per letter for the output state $\xi^{\widehat{A}_i \widehat{C}_i \widehat{B}_i R_i}$ is 
\begin{align*}
  D &= \tr \xi^{\widehat{A}_i \widehat{C}_i \widehat{B}_i R_i} \Delta \\
    &= 1 - \sum_{x_i} p(x_i) \tr \xi_{x_i ,x_i}^{\widehat{A}_i \widehat{C}_i \widehat{B}_i {R'}_i}\psi_{x_i}^{\widehat{A}_i \widehat{C}_i \widehat{B}_i {R'}_i} \\
    &= 1 - \sum_{x_i} p(x_i) F\left(\xi_{x_i ,x_i}^{\widehat{A}_i \widehat{C}_i \widehat{B}_i {R'}_i},\psi_{x_i}^{\widehat{A}_i \widehat{C}_i \widehat{B}_i {R'}_i}\right)^2.
\end{align*}
%
%
Again, up to a square this is the average fidelity considered in \cite{ISIT-paper,ZBK_PhD_thesis}, 
and it extends to the average-squared of per-copy  fidelity when the extended 
distortion operator of Eq.~(\ref{eq:Delta ensemble QSR}) is considered. This implies that in the limit of $D\to 0$,
the optimal compression rate of the ensemble source considering per-copy fidelity converges to 
$Q_{ea}(0)$. Therefore, by Theorem~\ref{thm:main} (as well as the results of \cite{ISIT-paper,ZBK_PhD_thesis}) we obtain that  $Q_{ea}(0)=\overline{Q}(0)$.

Now, we define a new single-letter function $K(D)$, 
which then we use to obtain simplified rate lower bounds that are easier to analyze. 

\medskip
\begin{definition}
  \label{def:K_epsilon}
  For 
  $\omega^{ACBR'X}=\sum_x p(x) \proj{\psi_x}^{ACBR'}\otimes \proj{x}^{X}$ a state and $D \geq 0$ define:
  \begin{align*}
    K(D) &:= \sup \frac12 I(W:X|\hat{C})_{\sigma} 
                                \text{ over isometries } \\
                       &\phantom{=====}
                        U: AC \rightarrow Z \hat{C} W \text{ and } 
                        \widetilde{U}:ZB \rightarrow \hat{A}\hat{B}V \text{ s.t.} \\
                       &\phantom{=====}
                        \sum_x p(x) F\left( \psi_x^{ACBR'},\tau_x^{\hat{A} \hat{C} \hat{B} R'}\right)^2  
                        \geq 1- D, 
  \end{align*}
  where 
  \begin{align*}
  \sigma^{Z\hat{C}WBR'X}
     &:= (U\otimes \1_{BR'X})\omega^{ACBR'X} (U\otimes \1_{BR'X})^{\dagger} \\
      & = \sum_x p(x) \proj{\sigma_x}^{Z\hat{C}WBR'}\otimes \proj{x}^{X}, \\
  \tau^{\hat{A}\hat{C}\hat{B}WVR'X}
     &:= (\widetilde{U}\otimes \1_{\hat{C}WR'X}) \sigma^{Z\hat{C}WBR'X}  
         (\widetilde{U}\otimes \1_{\hat{C}WR'X})^{\dagger}\\
       &= \sum_x p(x) \proj{\tau_x}^{\hat{A}\hat{C}\hat{B}WVR'}\otimes \proj{x}^{X}, \\
  \tau^{\hat{A}\hat{C}\hat{B}R'X}
     &:= \Tr_{VW} \tau^{\hat{A}\hat{C}\hat{B}WVR'X},\\
      \tau_x^{\hat{A}\hat{C}\hat{B}R'}
      &:=\Tr_{VW} \proj{\tau_x}^{\hat{A}\hat{C}\hat{B}WVR'}.
  \end{align*}
  Moreover, define $\overline{K}(0):=\lim_{D \to 0+} K(D)$.
\end{definition}

\medskip
\begin{remark}
Definition~\ref{def:K_epsilon} directly implies that $K(0) \leq \overline{K}(0)$ 
because $K(D)$ is a non-decreasing function of $D$. 
Furthermore, $K(0)$ can be strictly positive, for example, for a source 
with trivial system $C$ where $\psi_{x}^{A}\psi_{x'}^{A}=0$ holds for $x\neq x'$ for all $x,x'$, 
we obtain $K(0)=S(X)$.
This follows because Alice can measure her system and obtain the value of $X$ and 
then copy this classical information to the register $W$. 
\end{remark}

\medskip
\begin{lemma}
\label{lemma: lower bound on Q_tilde(0)}
For the source $\ket{\psi}^{ACBR} = \sum_x \sqrt{p(x)} \ket{\psi_x}^{ACBR'}\ket{x}^X$ and the distortion operator of Eq.~(\ref{eq:Delta ensemble QSR}), the rate $\overline{Q}(0)$ is lower bounded as:
\begin{align*}
    \overline{Q}(0) &\!\geq\! \frac{1}{2} \left(S(A|B)_{\psi}+S(A|C)_{\psi} \right) -\overline{K}(0) \\
   &\!=\!\frac{1}{2}I(A:R|B)_{\psi}-\overline{K}(0),
\end{align*}
where the above conditional mutual information is precisely the communication rate 
of QSR for the pure source $\ket{\psi}^{ACBR}$.
Moreover, if system $C$ is trivial, then $\overline{Q}(0)=\frac{1}{2}I(A:R|B)_{\psi}-\overline{K}(0)$.
\end{lemma}

\medskip
The slightly lengthy proof of this lemma is found in Appendix A.
We use it to simplify the rate expressions in important special cases. 

\medskip
\begin{definition}[{Barnum~\emph{et~al.}~\cite{Barnum2001_2}}]
\label{def:reducibility  QSR ensemble}
  An ensemble of pure states 
  $\cE=\{p(x),\proj{\psi_x}^{ACBR'} \}_{x\in \mathcal{X}}$ 
  is called \emph{reducible} if its states fall into two or more orthogonal subspaces.
  Otherwise the ensemble $\cE$ is called \emph{irreducible}.
  We apply the same terminology to the source state $\omega^{ACBR'X}$.
\end{definition}

\medskip
\begin{proposition}
\label{cor:irreducible}
For an irreducible source, $\overline{K}(0)=K(0)=0$. Hence, the 
optimal compression rate considering per-copy fidelity is
\begin{align*}
  \overline{Q}(0)=\frac{1}{2}I(A:R'XX'|B)_{\omega}=\frac{1}{2}I(A:R|B)_{\psi}.
\end{align*}
\end{proposition}

\begin{proof}
Consider the following mutual information
\begin{align*}
  \sup I(E:X|\hat{C})_{\nu} &
                  \text{ over isometries } 
                  U: ACB \rightarrow \hat{A} \hat{C} \hat{B}  E  \text{ s.t.} \\ &\sum_x p(x) F( \psi_x^{ACBR'},\nu_x^{\hat{A} \hat{C} \hat{B} R'})^2  \geq 1- D,
\end{align*}
where the state $\nu^{\hat{A} \hat{C} \hat{B} R'X}$ is the output state after applying the isometry $U$ on the input systems. In fact the isometries and the environments in Definition~\ref{def:K_epsilon} are respectively special cases of the above isometry and the environment $E$ in the above optimization. Therefore, the mutual information of Definition~\ref{def:K_epsilon} is bounded as
\begin{align}\label{eq:E vs VW}
    I(W:X|\widehat{C})_{\tau} \leq I(WV:X|\widehat{C})_{\tau} 
    \leq I(E:X|\hat{C})_{\nu} . 
\end{align}
Furthermore, for $D=0$ we obtain
\begin{align*}
  I(E:X|\hat{C})_{\nu} &\leq  I(E:X\hat{C})_{\nu} \\
                       &=     I(E:X)_{\nu}+I(E:\hat{C}|X)_{\nu}=I(E:X)_{\nu},
\end{align*}
where the last equality follows because for $D=0$ the environment $E$ and decoded system $\hat{C}$ are decoupled given $X$ (see Appendix B).
For irreducible sources 
the mutual information $I(E:X)_{\nu}$ is zero which follows from the detailed 
discussion of \cite[p.~2028]{Barnum2001_2}. 
In the limit $D \to 0$, the value of the optimization converges to  
its value at $D=0$ which follows from the fact that  
the fidelity  and the conditional mutual information are  continuous functions of
CPTP maps, and the domain of the  optimization is a compact set. 
Therefore, from Eq.~(\ref{eq:E vs VW}) we conclude that $I(W:X|\widehat{C})_{\tau}=0$.

The above proves $Q'(0) \geq \overline{Q}(0) \geq \frac{1}{2}I(A:R'XX'|B)_{\omega}$.
Also, by definition we have $Q'(0) \leq  \frac{1}{2}I(A:R'XX'|B)_{\omega}$. Therefore, $Q'(0) =  \overline{Q}(0) = \frac{1}{2}I(A:R'XX'|B)_{\omega}$. 
\end{proof}

\medskip
\begin{corollary}
\label{cor:converse for modified Schumacher ensemble}
The compression rate of the modified source defined in Eq.~(\ref{modified Schumacher source}) is bounded as follows
\begin{align*}
    \overline{Q}(0) &\!\geq\! \frac{1}{2} \left(S(A)_{\omega}+S(A|Y)_{\omega} \right).
\end{align*}
\end{corollary}
\begin{proof}
By Lemma~\ref{lemma: lower bound on Q_tilde(0)}, the first inequality below holds:
\begin{align}
    \overline{Q}(0) &\geq S(A)_{\omega} - \frac12 I(W:X)_\sigma \nonumber \\
    &=  \frac12 (S(A)_{\omega}+S(AY)_{\omega}) - \frac12 I(W:X)_\sigma \nonumber \\
    &=  \frac12 (S(A)_{\omega}+S(AY)_{\omega}) - \frac12 I(WY:X)_\sigma \nonumber \\
    &=  \frac12 (S(A)_{\omega}+S(AY)_{\omega}) - \frac12 I(Y:X)_\omega- \frac12 I(W:X|Y)_\sigma \nonumber \\
    &=  \frac12 (S(A)_{\omega}+S(AY)_{\omega}) - \frac12 I(Y:X)_\omega \nonumber \\
    &=  \frac12 (S(A)_{\omega}+S(A|Y)_{\omega}), \nonumber
\end{align}
where the second line follows because the information of the orthogonal
subspaces can obtained by an isometry on system $A$. The third line holds since $Y$ can be copied to the environment system. The penultimate line follows
from Proposition~\ref{cor:irreducible} because conditioned on $Y$, the source is irreducible.
The last line follows because $S(Y|X)=0$.
\end{proof}

\medskip
\begin{definition}
An ensemble of pure states $\cE=\{p(x),\proj{\psi_x}^{ACBR'} \}_{x\in \mathcal{X}}$
is called a \emph{generic source} if there is at least one $x$ 
for which the reduced state $\psi_x^{ACB}= \Tr_{R'} \ketbra{\psi_x}{\psi_x}^{ACBR'}$ 
has full support on $ACB$. 
\end{definition}

\medskip
\begin{proposition}
\label{prop:generic sources}
For generic sources, $\overline{K}(0)=K(0)=0$. Hence, the 
optimal compression rate considering per-copy fidelity is
\begin{align*}
  \overline{Q}(0)=\frac{1}{2}I(A:R'XX'|B)_{\omega}=\frac{1}{2}I(A:R|B)_{\psi}.
\end{align*}
\end{proposition}
We give the proof of this proposition in Appendix C. 
We note that in the case treated in this proposition, continuity of $Q_{ea}(D)$ 
and $Q(D)$ is guaranteed: it is anyway given at all $D>0$, and the above
result shows it holds also at $D=0$.

\section{Discussion}
\label{sec:discussion}
We consider an entanglement-assisted rate-distortion problem with side information systems at the encoder and decoder side where the distortion measure is a general convex and continuous function of source states.
We show that the optimal rate-distortion function is equal to the single-letter function
$Q'(D)$ for $D>D_0$ and $\lim_{D\to D_0}Q'(D)$ for $D=D_0$, where $D_0$ is minimal distortion.
Furthermore, we show that this is a convex and continuous function for $D>D_0$. Despite being  single letter, computing $Q'(D)$ potentially involves unbounded optimisation since a priori there is not a dimension bound on system $Z$. Therefore, we cannot apply compactness arguments to show that it is continuous at $D=D_0$.

We subsequently appply this general theory with specific distortion operators to study various source coding problems with per-copy fidelity criteria. We consider both pure and ensemble source models of Schumacher's compression and quantum state redistribution, and argue that we can always define quantum sources as pure states
and adjust the distortion operator accordingly to impose entanglement fidelity or ensemble fidelity as the decodability criterion. 
Therefore, we derive the optimal entanglement-assisted compression rates for Schumacher and QSR
sources with entanglement and ensemble fidelity.
For both Schumacher models and also pure QSR these rates are equal to the rates considering block fidelity. The ensemble QSR with block fidelity is studied in \cite{ISIT-paper, ZBK_PhD_thesis} where the converse is equal to $\overline{Q}(0)=\lim_{D\to D_0}Q'(D)$.
The rate $Q'(0)$ is shown to be achievable, and it  would only match with the converse 
if the function $Q'(D)$ is continuous at $D=0$.

To analyse the distortion measure for vanishing $D$, we find a lower bound on $\overline{Q}(0)$ in terms 
of the limit of another function at $D=0$, i.e. $\overline{K}(0)$. Despite the fact that
computing $\overline{K}(0)$ might involve unbounded optimization as well, it is sometimes
easier to analyse. In particular, we show that $\overline{K}(0)=0$ for irreducible and generic sources. This implies that for these sources both ensemble and entanglement fidelity lead to the same compression rate, i.e. the rate of pure QSR source.

Finally, recall that in our definition of the rate-distortion task we have assumed 
that the encoder and decoder share free entanglement. This was motivated so as 
to make a smoother connection to QSR.
However, it is not known whether the pre-shared entanglement is always necessary to achieve 
the corresponding quantum rates. There are certainly cases where QSR does not require prior 
entanglement, such as when Alice's side information $C$ is trivial, which would carry
over to our setting whenever $\overline{K}(0)=K(0)=0$, for instance for an irreducible
ensemble.
More generally, in future work we plan to consider the trade-off between the quantum 
rates and the entanglement rate.


\appendix

\section*{A. Proof of Lemma\ref{lemma: lower bound on Q_tilde(0)}}
For the pure source $\ket{\psi}^{ACBR}= \sum_x \sqrt{p(x)} \ket{\psi_x}^{ACBR'}\ket{x}^X$ and the distortion operator of Eq.~(\ref{eq:Delta ensemble QSR}), let $\cE_0$ and $\cD_0$ be the CPTP maps realizing the
infimum of $\frac12 I(Z:R|B)_{\varphi}$, in the definition of Eq.~(\ref{eq:rd-function-almost}). Moreover, let $U_{\cE_0}: AC \hookrightarrow Z \widehat{C} W$ and $U_{\cD_0}: ZB \hookrightarrow  \widehat{A} \widehat{B} V$ denote respectively the Stinespring isometries of $\cE_0$ and $\cD_0$. Then, the states after
applying the isometries are
\begin{align}
\ket{\varphi}^{Z\widehat{C} B W R}
  &= (U_{\cE_0} \ox \1_{BR}) \ket{\psi}^{ACBR},   \nonumber \\
\ket{\xi}^{\widehat{A} \widehat{C} \widehat{B}  WV R}
  &= (U_{\cD_0} \ox \1_{\widehat{C} W R}) \ket{\varphi}^{Z\widehat{C} B W R}. \nonumber
\end{align}
Now, let $\ket{\omega}^{ACBR'XX'}= \sum_x \sqrt{p(x)} \ket{\psi_x}^{ACBR'}\otimes \ket{x}^{X} \otimes \ket{x}^{X'}$ be the purification of the state $\omega^{ACBR'X}$ in Definition~\ref{def:K_epsilon} and define the following states:
\begin{align}\label{sigma tau pure states}
\ket{\sigma}^{Z\widehat{C} B W R'XX'}
  &:= (U_{\cE_0} \ox \1_{BR'XX'}) \ket{\omega}^{ACBR'XX'},   \nonumber \\
\ket{\tau}^{\widehat{A} \widehat{C} \widehat{B}  WV R'XX'}
  &:= (U_{\cD_0} \ox \1_{\widehat{C} W R'XX'}) \ket{\sigma}^{Z\widehat{C} B W R'XX'}. 
\end{align}
Notice that $\frac12 I(Z:R|B)_{\varphi}= \frac12 I(Z:R'XX'|B)_{\sigma}$. In what follows, we
establish lower bounds on $I(Z:R'XX'|B)_{\sigma}$. Namely, 
\begin{align}
 \label{eq:1}
 I&(Z:R'XX'|B)_{\sigma} \nonumber\\
  &= S(ZB)_{\sigma}-S(B)_{\sigma}-S(ZBR'XX')_{\sigma}+S(BR'XX')_{\sigma} \nonumber \\ 
  &= S(\widehat{C} W R'XX')_{\sigma}-S(B)_{\sigma}-S(\widehat{C} W)_{\sigma}+S(AC)_{\sigma} \nonumber \\
  &= S(AB)_{\omega}-S(AB)_{\omega}+S(C)_{\omega}-S(C)_{\omega} \nonumber\\
  &\phantom{==} 
     +S(\widehat{C} W R'XX')_{\sigma}-S(B)_{\sigma}-S(\widehat{C} W)_{\sigma}+S(AC)_{\sigma} \nonumber \\
  &= S(A|B)_{\omega}+S(A|C)_{\omega}-S(AB)_{\omega}+S(C)_{\omega} \nonumber\\
  &\phantom{==}
     +S(\widehat{C} W R'XX')_{\sigma}-S(\widehat{C} W)_{\sigma} \nonumber \\
  &= I(A:R'XX'|B)_{\omega} \nonumber\\
  &\phantom{==}
     -S(AB)_{\omega}+S(C)_{\omega}+S(\widehat{C} W R'XX')_{\sigma}-S(\widehat{C} W)_{\sigma}, 
\end{align}
where the first line follows by the definition of the conditional mutual information. 
The second line follows because the state  $\ket{\sigma}^{Z\widehat{C} B W R'XX'}$ is pure;
this simply implies that for example $S(ZB)_{\sigma}=S(\widehat{C} W R'XX')_{\sigma}$.
The fourth line follows by the definition of the quantum conditional entropy.
The last line follows since the state $\ket{\omega}^{ACBR'XX'}$ is pure. Also, notice that
$I(A:R'XX'|B)_{\omega}=I(A:R|B)_{\psi}$. 
We now focus on the last four terms in the last line of Eq.~(\ref{eq:1}), 
and rewrite their sum as

\begin{align}
 -S&(AB)_{\omega}+S(C)_{\omega}+S(\widehat{C} W R'XX')_{\sigma}-S(\widehat{C} W)_{\sigma} \nonumber\\
   &\phantom{=}
    = -S(AB)_{\omega}+S(C)_{\omega}+S(\widehat{A}\widehat{B} V)_{\tau}-S(\widehat{C} W)_{\sigma}. \label{eq:a1}
\end{align}
This holds because the state $\ket{\tau}^{\widehat{A} \widehat{C} \widehat{B}  WV R'XX'}$ is pure.
Now, we lower-bound the latter expression in Eq.~(\ref{eq:a1}), as follows: 
\begin{align}
  -S&(AB)_{\omega}+S(C)_{\omega}+S(\widehat{A}\widehat{B} V)_{\tau}-S(\widehat{C} W)_{\sigma} \nonumber\\
    &\geq -S(AB)_{\tau}+S(\widehat{C})_{\sigma}+S(\widehat{A}\widehat{B} V)_{\tau}-S(\widehat{C} W)_{\sigma} \nonumber\\
    &\phantom{============}
     -2\sqrt{2D} \log |C| - h\left(\sqrt{2D}\right) \label{eq:a2} \\
    &= -S(AB)_{\tau}+S(\widehat{A}\widehat{B} V)_{\tau}-S(W|\widehat{C})_{\sigma} \nonumber\\
    &\phantom{============}
     -2\sqrt{2D} \log |C| - h\left(\sqrt{2D}\right) \label{eq:a3} \\
    &\geq -S(\widehat{A}\widehat{B})_{\tau}+S(\widehat{A}\widehat{B} V)_{\tau}-S(W|\widehat{C} )_{\sigma} \nonumber\\
    &\phantom{=======}
     -2\sqrt{2D} \log |A| |C| |B| - 2h\left(\sqrt{2D}\right) \label{eq:a4}\\
    &= S(V|\widehat{A}\widehat{B} )_{\tau}-S(W|\widehat{C} )_{\sigma} \nonumber\\
    &\phantom{=======}
     -2\sqrt{2D} \log |A| |C| |B| - 2h\left(\sqrt{2D}\right) \nonumber\\
    &\geq S(V|\widehat{A}\widehat{B} X)_{\tau}-S(W|\widehat{C} )_{\sigma} \nonumber\\
    &\phantom{=======}
     -2\sqrt{2D} \log |A| |C| |B| - 2h\left(\sqrt{2D}\right) \label{eq:a5} \\
    &= S(\widehat{A}\widehat{B}V X )_{\tau}-S(\widehat{A}\widehat{B} X )_{\tau}-S(W|\widehat{C} )_{\sigma} \nonumber\\
    &\phantom{=======}
     -2\sqrt{2D} \log |A| |C| |B| - 2h\left(\sqrt{2D}\right) \nonumber \\
    &\geq S(\widehat{A}\widehat{B}V X )_{\tau}-S(AB X )_{\omega}-S(W|\widehat{C} )_{\sigma} \nonumber\\
    &\phantom{====}
     -2\sqrt{2D} \log |A|^2 |C| |B|^2 |X| -3h\left(\sqrt{2D}\right) \label{eq:a6} \\
    &= S(\widehat{A}\widehat{B}V X )_{\tau}-S(CR'X )_{\omega}-S(W|\widehat{C} )_{\sigma} \nonumber\\
    &\phantom{====}
     -2\sqrt{2D} \log |A|^2 |C| |B|^2 |X| - 3h\left(\sqrt{2D}\right) \label{eq:a7} \\
    &= S(\widehat{A}\widehat{B}V X )_{\tau}-S(\widehat{C}R'X )_{\sigma}-S(W|\widehat{C} )_{\sigma} \nonumber\\
    &\phantom{===}
     -2\sqrt{2D} \log |A|^2 |C|^2 |B|^2 |X|^2 |R'| - 4h\left(\sqrt{2D}\right) \label{eq:a8} \\
    &= S(\widehat{C}W R'X )_{\sigma}-S(\widehat{C}R'X )_{\sigma}-S(W|\widehat{C} )_{\sigma} \nonumber\\
    &\phantom{===}
     -2\sqrt{2D} \log |A|^2 |C|^2 |B|^2 |X|^2 |R'| - 4h\left(\sqrt{2D}\right) \label{eq:a9} \\
    &= -I(W:R'X|\widehat{C})_{\sigma} \nonumber\\
    &\phantom{===}
     -2\sqrt{2D} \log |A|^2 |C|^2 |B|^2 |X|^2 |R'| - 4h\left(\sqrt{2D}\right) \nonumber \\
    &= -I(W:X|\widehat{C})_{\sigma}-I(W:R'|\widehat{C}X)_{\sigma} \nonumber\\
    &\phantom{===}
     -2\sqrt{2D} \log |A|^2 |C|^2 |B|^2 |X|^2 |R'| - 4h\left(\sqrt{2D}\right) \label{eq:a10} \\
    &\geq -2K(D)-I(W:R'|\widehat{C}X)_{\sigma} \nonumber\\
    &\phantom{===}
     -2\sqrt{2D} \log |A|^2 |C|^2 |B|^2 |X|^2 |R'| - 4h\left(\sqrt{2D}\right) \label{eq:a11},
\end{align}
where Eq.~(\ref{eq:a2}) follows from the fidelity criterion in Definition~\ref{def:K_epsilon}: the 
output state on the system $\widehat{C}$ is $2\sqrt{2D}$-close 
to the original state $C$ in trace norm; then the inequality follows 
by applying the Fannes-Audenaert inequality 
\cite{Fannes1973,Audenaert2007}.
Eq.~(\ref{eq:a3}) follows due to definition of the quantum conditional entropy.
Eq.~(\ref{eq:a4}) follows from the fidelity criterion in Definition~\ref{def:K_epsilon}: the 
output state on the system $\widehat{A}\widehat{B}$ is $2\sqrt{2D}$-close 
to the original state $AB$ in trace norm; then the inequality follows 
by applying the Fannes-Audenaert inequality.
Eq.~(\ref{eq:a5}) follows from the subadditivity of the entropy.
Eq.~(\ref{eq:a6}) follows from the fidelity criterion in Definition~\ref{def:K_epsilon}: the 
output state on the system $\widehat{A}\widehat{B} X$ is $2\sqrt{2D}$-close 
to the original state $ABX$ in trace norm; then the inequality follows 
by applying the Fannes-Audenaert inequality.
Eq.~(\ref{eq:a7}) follows because the state $\ket{\omega}^{ACB R'XX'}$ is pure.
Eq.~(\ref{eq:a8}) follows from the fidelity criterion in Definition~\ref{def:K_epsilon}: the 
output state on the system $\widehat{C}R'X$ is $2\sqrt{2D}$-close 
to the original state $C R' X$ in trace norm; then the inequality follows 
by applying the Fannes-Audenaert inequality.
Eq.~(\ref{eq:a9}) follows because the state $\ket{\tau}^{\widehat{A} \widehat{C} \widehat{B}  WV R'XX'}$ is pure.
Eq.~(\ref{eq:a10}) follows from the chain rule.
Eq.~(\ref{eq:a11}) follows from the definition of $K(D)$. 

From Eqs.~(\ref{eq:1}) and (\ref{eq:a11}) we now obtain
\begin{align}
 \frac12 I(Z:R|B)_{\varphi} 
    &= \frac12 I(Z:R'XX'|B)_{\sigma} \nonumber\\
    &\!\!\!\!\!\!\!\!\!\!\!\!\!\!\!
     \geq \frac12 I(A:R'XX'|B)_{\omega} -K(D)-\frac12I(W:R'|\widehat{C}X)_{\sigma} \nonumber\\
    &
          -\sqrt{\frac{D}{2}} \log |A|^2 |C|^2 |B|^2 |X|^2 |R'| - 2h\left(\sqrt{2D}\right). \nonumber
\end{align}
In Section B of the Appendix we prove the \emph{the decoupling condition}:
$\lim_{D \to 0} I(W:R'|\widehat{C}X)_{\sigma}= \lim_{D \to 0} I(W:R'|X)_{\sigma}=0$. 
Therefore, in the limit of $D \to 0$ the inequality of the lemma follows
\begin{align}
  \frac12 I(Z:R|B)_{\varphi} 
      &\geq \frac12 I(A:R'XX'|B)_{\omega} - \overline{K}(0). \nonumber\\
      &=    \frac12 I(A:R|B)_{\psi} - \overline{K}(0). \nonumber
\end{align}
\medskip

Finally, we prove the last statement of the lemma, i.e.~for trivial systems $C$ and $\widehat{C}$, 
$\frac12 I(Z:R|B)_{\varphi} 
 = \frac12 I(A:R|B)_{\psi} -\overline{K}(0)$. 
From Eq.~(\ref{eq:1}), we have
\begin{align}
  \label{eq:c trivial 1}
  I(Z:R'XX'|B)_{\sigma}
    &=I(A:R'XX'|B)_{\omega}-S(AB)_{\omega} \nonumber\\
    &\phantom{===}
     +S( W R'XX')_{\sigma}-S( W)_{\sigma}\nonumber \\
    &\!\!\!\!\!\!
     =    I(A:R'XX'|B)_{\omega}-S(R'XX')_{\omega} \nonumber\\
    &\phantom{===}
     +S( W R'XX')_{\sigma}-S( W)_{\sigma}\nonumber \\
    &\!\!\!\!\!\!
     =    I(A:R'XX'|B)_{\omega}-I( W : R'XX')_{\sigma}\nonumber \\
    &\!\!\!\!\!\!
     \leq I(A:R'XX'|B)_{\omega}-I( W : R'X)_{\sigma},
\end{align}
where the second line follows because for trivial $C$ the state on systems $ABR'XX'$ is pure.
The last line follows from data processing inequality. Also, from Eq.~(\ref{eq:a10}) we have
\begin{align} 
  \label{eq:c trivial 2}
  I(Z:R'XX'|B)_{\sigma}
      &\geq I(A:R'XX'|B)_{\omega} -I(W:R'X)_{\sigma} \nonumber\\ 
      &\!\!\!\!\!\!\!\!\!
       -\sqrt{2D} \log |A|^2 |C|^2 |B|^2 |X|^2 |R'|  - 4h\left(\sqrt{2D}\right).
\end{align}
The decoupling condition (see Section B of the Appendix), with Eqs.~(\ref{eq:c trivial 1}) and 
(\ref{eq:c trivial 2}) imply that in the limit of $D \to 0$,
\begin{align} 
     \frac12 I(Z:R|B)_{\varphi} &= \frac12 I(Z:R'XX'|B)_{\sigma} \nonumber \\
     &= \frac12 I(A:R'XX'|B)_{\omega} -\frac12 I(W:X)_{\sigma} \nonumber \\
     &=\frac12  I(A:R|B)_{\psi} -\frac12 I(W:X)_{\sigma}. \nonumber
\end{align}
Notice that since the term $I(A:R|B)_{\psi}$ is constant, taking the infimum of $\frac12 I(Z:R|B)_{\varphi}$ is equivalent to taking the supremum of $I(W:X)_{\sigma}$, therefore, the lemma follows.

\section*{B. Decoupling condition} 
Here we prove that the conditional mutual information of 
Eq.~(\ref{eq:a11}) vanishes in the limit of $D\rightarrow 0$:
$$\lim_{D \to 0} I(W:R'|\widehat{C}X)_{\sigma}= \lim_{D \to 0} I(W:R'|X)_{\sigma}=0.$$

Consider the following reduced states of the states defined in Eq.~(\ref{sigma tau pure states}):

\begin{align}
\sigma^{Z\widehat{C} B W R'X}
  &= \sum_x p(x) \proj{\sigma_x}^{Z\widehat{C} B W R'} \ox \proj{x}^X, \nonumber \\
\tau^{\widehat{A} \widehat{C} \widehat{B} WV R'X}
  &= \sum_x p(x) \proj{\tau_x}^{\widehat{A} \widehat{C} \widehat{B} WV R'} \ox \proj{x}^X. \nonumber 
\end{align}
The fidelity criterion of Definition~\ref{def:K_epsilon} implies the following,
\begin{align} 
  \label{eq:fidelity psi tau}
    1-D & \leq \sum_x p(x) F\left( \psi_x^{ACBR'},\tau_{x}^{\hat{A} \hat{C} \hat{B} R'}\right)^2 \nonumber \\
    &=\sum_x p(x)\bra{\psi_x} \tau_{x}^{\hat{A} \hat{C} \hat{B} R'} \ket{\psi_x}     \nonumber \\
    & \leq \sum_x p(x) \left\| \tau_{x}^{\hat{A} \hat{C} \hat{B} R'} \right\|, 
\end{align}
where $\|\cdot\|$ denotes the operator norm, i.e.~the maximum singular value. 
Now, consider the Schmidt decomposition of the state 
$\ket{\tau_x}^{\widehat{A} \widehat{C} \widehat{B} WV R'}$ with respect to the partition
$\widehat{A} \widehat{C} \widehat{B}R' : WV$, i.e.
\begin{align*}
\ket{\tau_x}^{\widehat{A} \widehat{C} \widehat{B} WV R'}
    = \sum_{i} \sqrt{\lambda_{x}(i)}\ket{v_{x}(i)}^{\widehat{A} \widehat{C} \widehat{B} R'}  \ket{w_{x}(i)}^{WV}. 
\end{align*}
The fidelity of Eq.~(\ref{eq:fidelity psi tau}) implies that the subsystems of the partition 
are almost decoupled on average:  
\begin{align}
  \label{eq:almost-pure WV}
  \sum_x p(x) &F\left( \tau_x^{\hat{A} \hat{C} \hat{B}WVR'},\tau_x^{\hat{A} \hat{C} \hat{B}R'} \ox \tau_x^{WV} \right)^2                                                              \nonumber\\
    &= \sum_{x} p(x) \sum_i \lambda_{x}(i)^{3}                                                 \nonumber\\
    &\geq \sum_{x} p(x) \left\| \tau_x^{\hat{A} \hat{C} \hat{B}R'} \right\|^{3}                        \nonumber\\
    &\geq \left( \sum_{x} p(x) \left\| \tau_x^{\hat{A} \hat{C} \hat{B}R'}\right\| \right)^{3} 
     \geq (1-D)^3 
     \geq 1 - 3D,  
\end{align}
where the penultimate line follows from the convexity of the function $x^3$. The last line is due to Eq.~(\ref{eq:fidelity psi tau}).
By the Alicki-Fannes inequality, this implies
\begin{align}
  I(WV:R'|\widehat{C}X)_\tau 
     &=    S(R'|\widehat{C}X)_\tau -S(R'|WV\widehat{C}X)_\tau \nonumber \\
     &\leq S(R'|\widehat{C}X)_\tau -S(R'|WV\widehat{C}X)_{\widehat{\tau}} \nonumber\\
     &\phantom{===}
            +4\sqrt{6D} \log |R'| + g\left(2\sqrt{6D}\right) \nonumber \\
     &=    4\sqrt{6D} \log |R'| + g\left(2\sqrt{6D}\right), \nonumber  
\end{align}
where the second line follows because $\widehat{\tau}$ is $2\sqrt{6D}$-close to state $\tau$ in trace norm where $\widehat{\tau}^{\hat{A} \hat{C} \hat{B}R' VW X}:=\sum_x p(x)\tau_x^{\hat{A} \hat{C} \hat{B}R'} \ox \tau_x^{WV} \ox \proj{x}^{X}$. The last line follows because $S(R'|WV\widehat{C}X)_{\widehat{\tau}}=S(R'|\widehat{C}X)_{\tau}$. Then, the decoupling follows
in the limit of $D \to 0$ since by data processing $I(W:R'|\widehat{C}X)_\tau \leq I(WV:R'|\widehat{C}X)_\tau$. We can similarly prove that $\lim_{D \to 0} I(W:R'|X)_\tau=0$.

\section*{C. Proof of Proposition~\ref{prop:generic sources}}
We denote the Stinespring isometries of CPTP maps $\cE_0$ and $\cD_0$ from 
Definition~\ref{def:K_epsilon} respectively by 
$U_{\cE_0}: AC \hookrightarrow Z \widehat{C} W$ 
and $U_{\cD_0}: ZB \hookrightarrow  \widehat{A} \widehat{B} V$.
For generic sources, we show that the environment systems $W$ and $V$ satisfy
 $\lim_{D \to 0} I(WV:X|\widehat{C})_\tau =0$.
Thus, we obtain
\begin{align}
    \overline{K}(0) = \lim_{D \to 0} I(W:X|\widehat{C})_\tau \leq \lim_{D \to 0} I(WV:X|\widehat{C})_\tau= 0. \nonumber
\end{align}
First, we note that the fidelity in Eq.~(\ref{eq:almost-pure WV}) implies the following 
\begin{align}
    \sum_x p(x) \left\| \tau_x^{\hat{A} \hat{C} \hat{B}WVR'}-\tau_x^{\hat{A} \hat{C} \hat{B}R'} \ox \tau_x^{WV} \right\|_1 \leq 2\sqrt{6D}. \nonumber
\end{align}
We also obtain the following bound by the definition of the state $\tau$ (Definition~\ref{def:K_epsilon}): 
\begin{align}
    \sum_x p(x) \left\| \tau_x^{\widehat{A} \widehat{C} \widehat{B}R'} \ox \tau_x^{WV} -\psi_x^{ACBR'}\ox \tau_x^{WV} \right\|_1 \leq 2\sqrt{2D}. \nonumber
\end{align}
By applying the triangle inequality to the above equations, we obtain 
\begin{align}
  \label{eq: tau omega trace distance}
  \sum_x p(x) \left\| \tau_x^{\widehat{A} \widehat{C} \widehat{B}WVR'}-\psi_x^{ACBR'} \ox \tau_x^{WV} \right\|_1 
    &\leq 2\left(\sqrt{6D}+\sqrt{2D}\right) \nonumber\\
    &\leq 8\sqrt{D}.
\end{align}
Since the source is generic, there is an $x$, say $x=0$, 
for which $\psi_0^{ACB}$ has full support on
$\mathcal{L}(\mathcal{H}_{ACB})$, i.e.~$\lambda_0:=\lambda_{\min}(\psi_0^{ACB})>0$. 
Therefore, for any $\ket{\psi_x}^{ACBR'}$ there is an operator $T_x$ acting on the reference system $R'$ such that
\begin{align*}
  \ket{\psi_x}^{ACBR'} = (\1_{ACB} \otimes T_x) \ket{\psi_0}^{ACBR'},  \nonumber
\end{align*}
and $\| T_x\| \leq \frac{1}{\sqrt{\lambda_0}}$ \cite{ZK_cqSW_2018}, 
where again $\|\cdot\|$ denotes the operator norm.
We can also rewrite the output state as follows:
\begin{align}
  \tau_x^{\widehat{A} \widehat{C} \widehat{B}WVR'}
    &= (U_{\cD_0}U_{\cE_0} \ox \1_{R'}) \psi_x^{ACBR'}(U_{\cD_0}U_{\cE_0} \ox \1_{R'})^{\dagger} \nonumber\\
    &= (U_{\cD_0}U_{\cE_0} \ox \1_{R'})  (\1_{ACB} \ox T_x) \psi_0^{ACBR'} \nonumber\\
    &\phantom{=============}
       (\1_{ACB} \ox T_x)^{\dagger} (U_{\cD_0}U_{\cE_0} \ox \1_{R'})^{\dagger} \nonumber \\
    &=  (\1_{\widehat{A} \widehat{C} \widehat{B}WV} \ox T_x)(U_{\cD_0}U_{\cE_0} \ox \1_{R'}) \psi_0^{ACBR'}  \nonumber\\
    &\phantom{===========}
       (U_{\cD_0}U_{\cE_0} \ox \1_{R'})^{\dagger}(\1_{\widehat{A} \widehat{C} \widehat{B}WV} \ox T_x)^{\dagger} \nonumber\\
    &=  (\1_{\widehat{A} \widehat{C} \widehat{B}WV} \ox T_x) \tau_0^{\widehat{A} \widehat{C} \widehat{B}WVR'} (\1_{\widehat{A} \widehat{C} \widehat{B}WV} \ox T_x)^{\dagger}. \nonumber
\end{align}
We now replace $\psi_x$ and $\tau_x$ 
with the above states to obtain the following:
\begin{align}
  \sum_x &p(x) \left\| \tau_x^{\widehat{A} \widehat{C} \widehat{B}WVR'}-\psi_x^{ACBR'} \ox \tau_0^{WV} \right\|_1 \nonumber \\
    &= \sum_x p(x) \left\| (\1 \ox T_x) \tau_0^{\widehat{A} \widehat{C} \widehat{B}WVR'} (\1 \ox T_x)^{\dagger} \right.\nonumber\\
    &\phantom{=========}
                   \left. -(\1 \ox T_x) \psi_0^{ACBR'} (\1 \ox T_x)^{\dagger} \ox \tau_0^{WV} \right\|_1\nonumber \\
    &= \sum_x p(x) \|T_x\|^2 \, \left\| \tau_0^{\widehat{A} \widehat{C} \widehat{B}WVR'} - \psi_0^{ACBR'} \ox \tau_0^{WV} \right\|_1\nonumber \\
    &\leq 2\left(\sqrt{6D}+\sqrt{2D}\right) \sum_x p(x) \|T_x\|^2 
     \leq \frac{8\sqrt{D}}{{\lambda_0}}, \nonumber
\end{align}
where the third line follows from Eq.~(\ref{eq: tau omega trace distance}).  
We use the above upper bound on the average distance between the reduced states $\tau_x^{\widehat{C} WV}$ and $\psi_x^{C} \ox \tau_0^{WV}$ to conclude that the environment\\

\noindent
systems $WV$ are decoupled from systems $\widehat{C} X$:
\begin{align}
   \frac{1}{2} &\Bigl\| \underbrace{ \sum_x  p(x)\tau_x^{\widehat{C}WV}\ox \proj{x}^X}_{\tau^{\widehat{C}WV X}}-\underbrace{\sum_x p(x)\omega_x^C \ox \tau_0^{WV} \ox \proj{x}^X}_{\zeta^{C WV X}} \Bigr\|_1 \nonumber\\
    &\phantom{=======}
     \leq \frac{1}{2} \sum_x p(x) \left\| \tau_x^{\widehat{C}WV}- \omega^{C}_x \ox \tau_0^{WV} \right\|_1 
     \leq \frac{8\sqrt{D}}{{\lambda_0}}=:\delta_D. \nonumber
\end{align}
By applying the Alicki-Fannes inequality in the form of Eq.~(\ref{thm:AFW}) to 
the above states, we obtain
\begin{align*}
  I&(WV:\widehat{C}X)_\tau = S(\widehat{C}X)_{\tau}-S(CX|WV)_{\zeta}+S(CX|WV)_{\zeta} \\
   &\phantom{===============================} -S(\widehat{C}X|WV)_{\tau} \\
   &=    S(\widehat{C}X)_{\tau}-S(CX)_{\omega} +S(CX|WV)_{\zeta}-S(\widehat{C}X|WV)_{\tau} \\
   &\leq S(\widehat{C}X|WV)_{\zeta}-S(\widehat{C}X|WV)_{\tau} \\
   &\phantom{====}
          +2\sqrt{2D} \log |C| |X| + h\left(\sqrt{2D}\right)  \\
   &\leq  2\delta_D \log |C| |X| + h(\delta_D) +2\sqrt{2D} \log |C| |X| + h\left(\sqrt{2D}\right), 
\end{align*}
where the third line follows from the fidelity criterion in Definition~\ref{def:K_epsilon}: the
output state on the system $\widehat{C}X$ is $2\sqrt{2D}$-close 
to the original state $CX$ in trace norm; then the inequality follows 
by applying the Fannes-Audenaert inequality. 
Therefore, we conclude that $ I(W:X|\widehat{C})_\tau \leq  I(WV:X|\widehat{C})_\tau \leq I(WV:X\widehat{C})_\tau$ which the latter vanishes for $D \to 0$.

\bibliographystyle{IEEEtran}


\end{document}